\documentclass[final,oneeqnum,onethmnum,onefignum,10pt]{siamltex}

\usepackage{amsmath}
\usepackage{amssymb}
\usepackage[pdftex]{color,graphicx}

\usepackage{hyperref}
\hypersetup{plainpages = false,
              breaklinks=true,
              a4paper=true,
              linktocpage,
              pdfauthor={Niko Brummer},
              pdfpagemode=None,
              pdfstartpage=1, 
              pdfstartview=FitH,  
              pdfkeywords={}}

\newcommand{\funcdef}[3]{#1:#2\to#3}
\def\wsf{v}
\def\Crho{\operatorname{C}_\rho}
\def\prop{\theta}
\newcommand{\set}[1]{\left\{#1\right\}}
\newcommand{\seq}[2]{#1_{1},#1_{2},\ldots,#1_{#2}}
\newcommand{\pvec}{\mathbf{p}}
\newcommand{\qvec}{\mathbf{q}}
\def\Obj{\mathcal{O}}
\def\nllv{\mathbf{w}}
\newcommand{\logit}{\operatorname{logit}}
\newcommand{\PAV}{\operatorname{PAV}}
\def\R{\mathbb{R}}
\def\Exs#1#2#3%
{%
  \mathbb{E}_{#1|#2}\bigl\{#3\bigr\}%
}

\title{The PAV Algorithm optimizes Binary Proper Scoring Rules}

\author{Niko Br\"ummer\footnotemark[2]\ \footnotemark[3] \and Johan du Preez\footnotemark[3]}
\begin{document}
\maketitle
\renewcommand{\thefootnote}{\fnsymbol{footnote}}
\footnotetext[2]{Spescom DataVoice, Box 582, Stellenbosch 7599, South Africa. Email: niko.brummer@gmail.com}
\footnotetext[3]{Digital Signal Processing Group, Department of Electrical and Electronic Engineering, University of Stellenbosch. Private Bag X1, Matieland, 7602, Stellenbosch, South Africa. Email: dupreez@dsp.sun.ac.za}
\renewcommand{\thefootnote}{\arabic{footnote}}

\begin{abstract}
There has been much recent interest in application of the \emph{pool-adjacent-violators (PAV) algorithm} for the purpose of \emph{calibrating} the probabilistic outputs of automatic pattern recognition and machine learning algorithms. Special cost functions, known as \emph{proper scoring rules} form natural objective functions to judge the goodness of such calibration. We show that for binary pattern classifiers, the non-parametric optimization of calibration, subject to a monotonicity constraint, can solved by PAV and that this solution is optimal for all regular binary proper scoring rules. This extends previous results which were limited to \emph{convex} binary proper scoring rules. We further show that this result holds not only for calibration of probabilities, but also for calibration of \emph{log-likelihood-ratios}, in which case optimality holds independently of the prior probabilities of the pattern classes.     
\end{abstract}

\begin{keywords} 
pool-adjacent-violators algorithm, proper scoring rule, calibration, isotonic regression 
\end{keywords}


\pagestyle{myheadings}
\thispagestyle{plain}
\markboth{N.Br\"ummer and J.du Preez}{PAV and Proper Scoring Rules}

\section{Introduction}
There has been much recent interest in using the \emph{pool-adjacent-violators}\footnote{a.k.a \emph{pair}-adjacent-violators} (PAV) algorithm for the purpose of \emph{calibration} of the outputs of machine learning or pattern recognition systems~\cite{zadrozny_elkan,csl,mizil_05,synergy,daniel,rocch}. Our contribution is to point out and prove some previously unpublished results concerning the optimality of using the PAV algorithm for such calibration. 

In the rest of the introduction, \S\ref{sec:calibration} defines calibration; \S\ref{sec:rbpsr} introduces \emph{regular binary proper scoring rules}, the class of objective functions which we use to judge the goodness of calibration; and~\S\ref{sec:smnpc} gives more specific details of how this calibration problem forms the non-parametric, monotonic optimization problem which is the subject of this paper.

The rest of the paper is organized as follows: In~\S\ref{sec:main_problem} we state the main optimization problem under discussion; \S\ref{sec:prev_work} summarizes previous work related to this problem; \S\ref{sec:proof}, the bulk of this paper, presents our proof that PAV solves this problem; and finally~\S\ref{sec:PAV-LLR} shows that the PAV can be adapted to a closely related calibration problem, which has the goal of assigning calibrated \emph{log-likelihood-ratios}, rather than probabilities. We conclude in~\S\ref{discussion} with a short discussion about applying PAV calibration in pattern recognition.

The results of this paper can be summarized as follows: The PAV algorithm, when used for supervised, monotonic, non-parametric calibration is (i) optimal for all regular binary proper scoring rules and is moreover (ii) optimal at any prior when calibrating log-likelihood-ratios. 

\subsection{Calibration}
\label{sec:calibration}
In this paper, we are interested in the calibration of \emph{binary} pattern classification systems which are designed to discriminate between two classes, by outputting a scalar \emph{confidence score}\footnote{The reader is cautioned not to confuse \emph{score} as defined here, with \emph{proper scoring rule} as defined in the next subsection.}. Let $x$ denote a to-be-classified input pattern\footnote{The nature of $x$ is unimportant here, it can be an image, a sound recording, a text document etc.}, which is known to belong to one of two classes: the \emph{target} class $\prop_1$, or the \emph{non-target} class $\prop_2$. The pattern classifier under consideration performs a mapping $x\mapsto s$, where $s$ is a real number, which we call the \emph{uncalibrated confidence score}. The only assumption that we make about $s$ is that it has the following \emph{sense}: \emph{The greater the score, the more it favours the target class---and the smaller, the more it favours the non-target class.}  

In order for the pattern classifier output to be more generally useful, it can be processed through a calibration transformation. We assume here that the calibrated output will be used to make a minimum-expected-cost Bayes decision~\cite{DeGroot,Wald}. This requires that the score be transformed to act as \emph{posterior probability} for the target class, given the score. We denote the transform of the uncalibrated score $s$ to calibrated target posterior thus: $s\mapsto P(\prop_1|s)$. In the first (and largest) part of this paper, we consider this calibration transformation as an atomic step and show in what sense the PAV algorithm is optimal for this transformation.

In most machine-learning contexts, it is assumed that the object of calibration is (as discussed above) to assign \emph{posterior probabilities}~\cite{platt,zadrozny_elkan,mizil_05}. However, the calibration of \emph{log-likelihood-ratios} may be more appropriate in some pattern recognition fields such as automatic speaker recognition~\cite{doddington,csl}. This is important in particular for \emph{forensic} speaker recognition, in cases where a Bayesian framework is used to represent the weight of the speech evidence in likelihood-ratio form~\cite{daniel}. With this purpose in mind, in~\S\ref{sec:PAV-LLR}, we decompose the transformation $s\mapsto P(\prop_1|s)$ into two consecutive steps, thus: $s\mapsto \log\frac{P(s|\prop_1)}{P(s|\prop_2)} \mapsto P(\prop_1|s)$, where the intermediate quantity is known as the \emph{log-likelihood-ratio} for the target, relative to the non-target. The first stage, $s\mapsto \log\frac{P(s|\prop_1)}{P(s|\prop_2)}$, is now the calibration transform and it is performed by an adapted PAV algorithm (denoted PAV-LLR), while the second stage, $\log\frac{P(s|\prop_1)}{P(s|\prop_2)} \mapsto P(\prop_1|s)$, is just standard application of Bayes' rule. One of the advantages of this decomposition is that the log-likelihood-ratio is independent of $P(\prop_1)$, the prior probability for the target class---and that therefore the pattern classifier (which does $x\mapsto s$) and the calibrator (which does $s\mapsto\log\frac{P(s|\prop_1)}{P(s|\prop_2)}$) can both be independent of the prior. The target prior need only be available for the final step of applying Bayes' rule. Our important contribution here is to show that the PAV-LLR calibration is optimal \emph{independently} of the prior $P(\prop_1)$.

\subsection{Regular Binary Proper Scoring Rules}
\label{sec:rbpsr}
We have introduced calibration as a tool to map uncalibrated scores to posterior probabilities, which may then be used to make minimum-expected-cost Bayes decisions. We next ask how the quality of a given calibrator may be judged. Since the stated purpose of  calibration is to make cost-effective decisions, the goodness of calibration may indeed be judged by decision cost. For this purpose, we consider a class of special cost functions known as \emph{proper scoring rules} to quantify the cost-effective decision-making ability of posterior probabilities, see e.g.~\cite{Good,DeGroot,DeGFien,Dawid,Buja_degrees,Gneiting}, or our previous work~\cite{csl}. Since this paper is focused on the PAV algorithm, a detailed introduction to proper scoring rules is out of scope. Here we just need to define the class of regular binary proper scoring rules in a way that is convenient to our purposes. (Appendix~\ref{appendix} gives some notes to link this definition to previous work.) 

We define a \emph{regular binary proper scoring rule} (RBPSR) to be a function,  $\funcdef{\Crho}{\set{\prop_1,\prop_2}\times[0,1]}{[0,\infty]}$, such that
\begin{align}
\label{eq:c_rbpsr}
\Crho(\prop_1,q) & = \int_{q}^{1}\frac{1}{\eta}\rho(\eta) \,d\eta, & 
\Crho(\prop_2,q) & = \int_{0}^{q}\frac{1}{1-\eta}\rho(\eta) \,d\eta 
\end{align} 
for which the following conditions must hold:
\begin{romannum}
\item These integrals exist and are \emph{finite}, except\footnote{This exception accommodates cases like the logarithmic scoring rule, which is obtained at $\rho(\eta)=1$, see~\cite{Dawid,Gneiting}.} possibly for $\Crho(\prop_1,0)$ and $\Crho(\prop_2,1)$, which may assume the value $\infty$.
\item $\rho(\eta)$ is a probability distribution\footnote{It is easily shown that if $\rho(\eta)$ cannot be normalized (i.e. $\int_0^1 \rho(\eta)\,d\eta\to\infty$), then one or both of $\Crho(\prop_1,q)$ or $\Crho(\prop_2,q)$ must also be infinite for every value of $q$, so that a useful proper scoring rule is not obtained.} over $[0,1]$, i.e. $\rho(\eta)\ge0$ for $0\le\eta\le 1$, and $\int_0^1 \rho(\eta)\,d\eta = 1$. 
\end{romannum}
In other words the RBPSR's are a family of functions parametrized by $\rho$. If $\rho(\eta)>0$ almost everywhere, then the RBPSR is denoted \emph{strict}, otherwise it is \emph{non-strict}. We list some examples, which will be relevant later:
\begin{remunerate}
\item If $\rho(\eta)=\delta(\eta-\eta')$, where $\delta$ denotes Dirac-delta, then $\Crho(\cdot,q)$ represents the misclassification cost of making binary decisions by comparing probability $q$ to a threshold of $\eta'$. Note that this proper scoring rule is \emph{non-strict}. Moreover it is discontinuous and therefore \emph{not convex} as a function of $q$. This is but one example of many non-convex proper scoring rules. A more general example is obtained by convex combination\footnote{The $\alpha_i>0$ and sum to $1$.} of multiple Dirac-deltas: $\rho(\eta)=\sum_i \alpha_i\delta(\eta-\eta'_i)$.
\item If $\rho(\eta)=6\eta(1-\eta)$, then $\Crho$ is the  (\emph{strict}) quadratic\footnote{In this context the average of the Brier proper scoring is just a mean-squared-error.} proper scoring rule, also known as the Brier scoring rule~\cite{Brier}.
\item If $\rho(\eta)=1$, then $\Crho$ is the (\emph{strict}) logarithmic scoring rule, originally proposed by~\cite{Good}.  
\end{remunerate}
The salient property of a binary proper scoring rule is that for any $0\le p,q\le 1$, its expectations w.r.t $q$ are minimized at $q$, so that:  $q\Crho(\prop_1,q)+(1-q)\Crho(\prop_2,q) \le q\Crho(\prop_1,p)+(1-q)\Crho(\prop_2,p)$. For a strict RBPSR, this minimum is unique. We show below in lemma~\ref{lemma3} how this property derives from~\eqref{eq:c_rbpsr}.  
 
\subsection{Supervised, monotonic, non-parametric calibration}
\label{sec:smnpc}
We have thus far established that we want to find a calibration method to map scores to probabilities and that we then want to judge the goodness of these probabilities via RBPSR. We can now be more specific about the calibration problem that is optimally solvable by PAV:
\begin{remunerate}
\item Firstly, we constrain the calibration transformation $s\mapsto P(\prop_1|s)$ to be a \emph{monotonic non-decreasing} function: $\R\to[0,1]$. This is to preserve the above-defined \emph{sense} of the score $s$. This monotonicity constraint is discussed further in~\S\ref{discussion}. See also~\cite{csl,zadrozny_elkan,mizil_05,daniel}.  
\item Secondly, we assume that we are given a finite number, $T$, of \emph{trials}, for each of which the to-be-calibrated pattern classifier has produced a score. We denote these scores $s_1,s_2,\ldots s_T$. We need only to map each of these scores to a probability. In other words, we do not have to find the calibration function itself, we only have to \emph{non-parametrically} assign the $T$ function output values $p_1,p_2,\ldots,p_T$, while respecting the above monotonicity constraint. To simplify notation, we assume without loss of generality, that $s_1 \le s_2 \le \cdots \le s_T$. (In practice one has to sort the scores to make it so.) This now means that monotonicity is satisfied if $0\le p_1 \le p_2 \le \cdots \le p_T\le 1$. Notice that the input scores now only serve to define the order. Once this order is fixed, one does not need to refer back to the scores. The output probabilities can now be independently assigned, as long as they respect the above chain of inequalities.     
\item Finally, we assume that the problem is \emph{supervised}: For every one of the $T$ trials the true class is known and is denoted: $\ell_1,\ell_2,\ldots,\ell_T\in\set{\prop_1,\prop_2}$. This allows evaluation of the RBPSR for every trial $t$ as $\Crho(\ell_t,p_t)$. A weighted combination of the RBPSR costs for every trial can now be used as the objective function which needs to be be minimized.
\end{remunerate}
In summary the problem which is solved by PAV is that of finding $p_1,p_2,\ldots,p_T$, subject to the monotonicity constraints, so that the RBPSR objective is minimized. This problem is succinctly restated in the following section:

\section{Main optimization problem statement}
\label{sec:main_problem}
The problem of interest may be stated as follows:
\begin{remunerate}
\item We are given as input:
\begin{romannum}
	\item A sequence of $T$ indices, denoted $(1,T)=1,2,\ldots,T$ with a corresponding sequence of labels $\seq{\ell}{T}\in\set{\prop_1,\prop_2}$.
  \item A pair of positive weights, $v_1,v_2>0$. 
\end{romannum}
\item We use the notation $\wsf(\ell_t)$ to assign a weight to every index, by letting $\wsf(\prop_1)=v_1$ and $\wsf(\prop_2)=v_2$.
\item The problem is now to find the sequence of $T$ probabilities, denoted $\pvec_{1,T}=\seq{p}{T}$, which minimizes the following \emph{objective}:
\begin{align}
\label{eq:pav_obj}
\Obj_{1,T}(\pvec_{1,T}) &= \sum_{t=1}^T \wsf(\ell_t)\Crho(\ell_t,p_t),
\end{align} 
subject to the \emph{monotonicity constraint}:
\begin{equation}
\label{eq:mon_constr}
0 \le p_1 \le p_2 \le \cdots \le p_T\ \le 1
\end{equation}
We require the solution to hold (be a feasible minimum) \emph{simultaneously} for every RBPSR $\Crho$. We already know that if such a solution exists, it must be unique, because the original PAV algorithm as published in~\cite{ayer1955} in 1955, was shown to give a unique optimal solution for the special case of $\rho(\eta)=1$, for which $\bigl(\Crho(\prop_1,p),\Crho(\prop_2,p)\bigr)=\bigl(-\log(p),-\log(1-p)\bigr)$. See theorem~\ref{th:ayer} and corollary~\ref{corollary2} below for details. 
\end{remunerate}

\section{Relationship of our proof to previous work}
\label{sec:prev_work}
Although not stated explicitly in terms of a proper scoring rule, the first publication of the PAV algorithm~\cite{ayer1955}, was already proof that it optimized the logarithmic proper scoring rule. It is also known that PAV optimizes the quadratic (Brier) scoring rule~\cite{zadrozny_elkan}, and indeed that it optimizes combinations of more general convex functions~\cite{best,ahuja_orlin_01}. However as pointed out above, there are proper scoring rules that are not convex.

In our previous work~\cite{csl}, where we made use of calibration with the PAV algorithm, we did mention the same results presented here, but without proof. This paper therefore complements that work, by providing proofs.

We also note that independently, in~\cite{rocch}, it was stated ``it can be
proved that the same [PAV algorithm] is obtained when using
any proper scoring function'', but this was also without proof or further references\footnote{Notes to reviewers: Note 1: We contacted Fawcet and Niculescu-Mizil to ask if they had a proof. They replied that their statement was based on the assumption that proper scoring rules are convex, which by~\cite{best} is then optimized by PAV. Since we include here also non-convex proper scoring rules, our results are more general. Note 2: The paper~\cite{ubhaya} has the word `quasi-convex' in the title and employs the PAV algorithm for a solution. This could suggest that our problem was solved in that paper, but a different problem was solved there, namely: ``the approximation problem of fitting n data points by a quasi-convex function using the least squares distance function.''}.      

We construct a proof that the PAV algorithm solves the problem as stated in~\S\ref{sec:main_problem}, by roughly following the pattern of the unpublished document~\cite{ahuja_orlin_98}, where the optimality of PAV was proved for the case of \emph{strictly convex} cost functions. That proof is not applicable as is for our purposes, because as pointed out above, some RBPSR's are not convex. We will show however in lemma~\ref{lemma3} below, that all RBPSR's and their expectations are \emph{quasiconvex} and that the proof can be based on this quasiconvexity, rather than on convexity. Note that when working with convex cost functions, one can use the fact that positively weighted combinations of convex functions are also convex, but this is not true in general for quasiconvex functions. For our case it was therefore necessary to prove explicitly that expectations of RBPSR's are also quasiconvex. A further complication that we needed to address was that non-strict RBPSR's lead to unidirectional implications, in places where the strictly convex cost functions of the proof in~\cite{ahuja_orlin_98} gave \emph{if and only if} relationships. 

Finally, we note that although the more general case of PAV for non-strict convex cost functions was treated in~\cite{best}, we could not base our proof on theirs, because they used properties of convex functions, such as subgradients, which are not applicable to our quasiconvex RBPSR's.

\section{Proof of optimality of PAV} 
\label{sec:proof}
This section forms the bulk of this paper and is dedicated to prove that a version of the PAV algorithm solves the optimization problem stated in~\S\ref{sec:main_problem}.
 
\begin{figure}[!ht] 
\centerline{
  \includegraphics[
                   width=0.8\textwidth]
                   {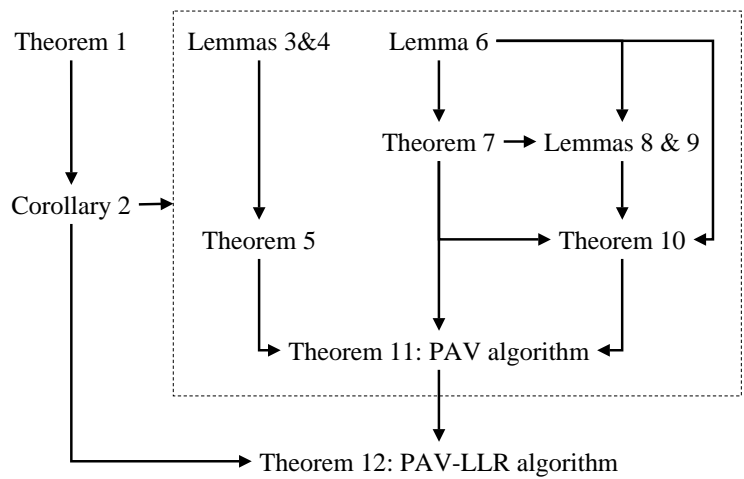}
}
\caption[PAV Proof]{Proof structure: PAV is optimal for all RBPSR's and PAV-LLR is optimal for all RBPSR's and priors.}
\label{fig:pav_proof}
\end{figure}

See figure~\ref{fig:pav_proof} for a roadmap of the proof: Theorem~\ref{th:ayer} and corollary~\ref{corollary2} give the closed-form solution for the logarithmic RBPSR. For the PAV algorithm, we use corollary~\ref{corollary2} just to show that there is a unique solution, but we re-use it later to prove the prior-independence of the PAV-LLR algorithm. Inside the dashed box, theorem~\ref{theorem1} shows how multiple optimal \emph{subproblem} solutions can constitute the optimal solution to the whole problem. Theorems~\ref{theorem2} and~\ref{theorem3} respectively show how to find and combine optimal subproblem solutions, so that the PAV algorithm can use them to meet the requirements of theorem~\ref{theorem1}. 
   
\subsection{Unique solution}
In this section, we use the work of Ayer et al, reproduced here as theorem~\ref{th:ayer}, to show via corollary~\ref{corollary2} that, if our problem does have a solution for every RBPSR, then it must be unique, because the special case of the logarithmic scoring rule (when $\rho(\eta)=1$) does have a unique solution.
\begin{theorem}[Ayer et al., 1955]
\label{th:ayer}
Given non-negative real numbers $a_t,b_t$, such that $a_t+b_t>0$ for every $t=1,2,\ldots,T$, the maximization of the objective $\Obj'_{1,T}(\pvec_{1,T})=\prod_{t=1}^T (p_t)^{a_t}(1-p_t)^{b_t}$, subject to the monotonicity constraint~\eqref{eq:mon_constr}, has the unique solution, $\pvec_{1,T}=\seq{p}{T}$, where:
\begin{equation}
\begin{split}
p_t &= \max_{1\le i\le t} \; \min_{t\le j\le T} r'_{i,j} \\
& = \min_{t\le j\le T} \max_{1\le i\le t} r'_{i,j},\\
\end{split}
\end{equation}
where
\begin{equation}
r'_{i,j} =\frac{\sum_{k=i}^j a_k}{\sum_{k=i}^j a_k+b_k}
\end{equation}
\end{theorem}
\begin{proof}
See\footnote{Available online (with open access) at \url{http://projecteuclid.org/euclid.aoms/1177728423}.}~\cite{ayer1955}, theorem 2.2 and its corollary 2.1. In that work, the monotonicity constraint was non-increasing, rather than the non-decreasing constraint~\eqref{eq:mon_constr} that we use here. The solution that they give therefore has to be transformed by letting the index $t$ go in reverse order, which means exchanging the roles of the subsequence endpoints $i,j$, which then has the result of exchanging the roles of $\max$ and $\min$ in the solution. 
\qquad\end{proof}

We now show that this theorem supplies the solution for the special case of the logarithmic RBPSR: 
\begin{corollary}
\label{corollary2}
If $\bigl(\Crho(\prop_1,p),\Crho(\prop_2,p)\bigr)=\bigl(-\log(p),-\log(1-p)\bigr)$, then the problem of minimizing objective~\eqref{eq:pav_obj}, subject to constraint~\eqref{eq:mon_constr}, has the unique solution, $\pvec_{1,T}=\seq{p}{T}$, where:  
\begin{equation}
\label{eq:pav_maxmin}
\begin{split}
p_t =\PAV_t\bigl((\seq{\ell}{T}),(v_1,v_2)\bigr) &= \max_{1\le i\le t} \; \min_{t\le j\le T} r_{i,j} \\
&= \min_{t\le j\le T} \; \max_{1\le i\le t} r_{i,j},\\
\end{split}
\end{equation}
where
\begin{equation}
r_{i,j} =\frac{m_{i,j}v_1}{m_{i,j}v_1+n_{i,j}v_2}
\end{equation}
where $m_{i,j}$ is the number of $\prop_1$-labels and $n_{i,j}$ the number of $\prop_2$-labels in subsequence $\ell_i,\ell_{i+1},\ldots,\ell_{j}$.
\end{corollary}
\begin{proof}
Observe that if we let
\begin{align*}
(a_t,b_t) &= 
  \begin{cases}
  (v_1,0),& \text{if $\ell_t=\prop_1$}, \\
  (0,v_2),& \text{if $\ell_t=\prop_2$},
  \end{cases} 
\end{align*}
then $r'_{i,j}=r_{i,j}$, so that  $\Obj'_{1,T}(\pvec_{1,T})=\operatorname{exp}\bigl(-\Obj_{1,T}(\pvec_{1,T})\bigr)$, so that the constrained maximization of theorem~\ref{th:ayer} and the constrained minimization of this corollary have the same solution.
\qquad\end{proof}

This corollary gives a closed-form solution, \eqref{eq:pav_maxmin}, to the problem, and from~\cite{ayer1955} we know that this is the same solution which is calculated by the iterative PAV algorithm\footnote{The PAV algorithm, if efficiently implemented, is known~\cite{pardalos,ahuja_orlin_01,synergy} to have \emph{linear} computational load (of order $T$), which is superior to a straight-forward implementation of the explicit form~\eqref{eq:pav_maxmin}.}. As noted above, it has so far~\cite{ayer1955,ahuja_orlin_98,best} only been shown that this solution is valid for logarithmic and other RBPSR's which have \emph{convex} expectations. In the following sections we show that this solution is also optimal for all other RBPSR's. 

\subsection{Decomposition into subproblems}
We need to consider \emph{subsequences} of $(1,T)$: For any $1\le i\le j\le T$, we denote as $(i,j)$ the subsequence of $(1,T)$ which starts at index $i$ and ends at index $j$. We may compute a partial objective function over a subsequence $(i,j)$ as:
\begin{equation}
\label{eq:partobj}
\Obj_{i,j}(\pvec_{i,j}) = \sum_{t=i}^j \wsf(\ell_t) \Crho(\ell_t,p_t).
\end{equation}
where $\pvec_{i,j}=p_i,p_{i+1},\ldots,p_j$. We can now define the \emph{subproblem} $(i,j)$ as the problem of minimizing $\Obj_{i,j}(\pvec_{i,j})$, simultaneosly for every RBPSR, and subject to the monotonicity constraint $0\le p_i\le p_{i+1}\le\cdots\le p_j\le1$. In what follows, we shall use the following notational conventions:
\begin{remunerate}
	\item The subproblem $(1,T)$ is equivalent to the original problem.
  \item We shall denote a subproblem solution, $\pvec_{i,j}$, as \emph{feasible} when the monotonicity constraint is met and \emph{non-feasible} otherwise.
  \item By \emph{subproblem solution} we mean just a sequence $\pvec_{i,j}$, feasible or not, such that $p_i,p_{i+1},\ldots,p_j\in[0,1]$.
  \item Since any subproblem is isomorphic to the original problem, corollary~\ref{corollary2} also shows that if\footnote{The object of this whole exercise is to prove that the optimal solution exists for every subproblem and is given by the PAV algorithm, but until we have proved this, we cannot assume that the optimal solution exists for every subproblem.} it has a feasible minimizing solution for every RBPSR, then that solution must be unique. Hence, by \emph{the optimal subproblem solution}, we mean the unique feasible solution that minimizes $\Obj_{i,j}(\cdot)$, for every RBPSR.
  \item By a \emph{partitioning} of the problem $(1,T)$ into a set, $\mathcal{S}$, of adjacent, non-overlapping subproblems, we mean that every index occurs exactly once in all of the subproblems, so that:
\begin{equation}
\label{eq:totalobjsum}
\Obj_{1,T}(\pvec_{1,T}) 
= \sum_{(i,j)\in\mathcal{S}} \Obj_{i,j}(\pvec_{i,j})
\end{equation}   
\end{remunerate}
Our first important step is to show with theorem~\ref{theorem1}, proved via lemmas~\ref{lemma1a} and~\ref{lemma2a}, how the optimal total solution may be constituted from optimal subproblem solutions:      
\begin{lemma}
\label{lemma1a}
For a given RBPSR and for a given partitioning, $\mathcal{S}$, of $(1,T)$ into subproblems, let:
\begin{remunerate}
	\item[(i)] $\pvec^*_{1,T}=\seq{p^*}{T}$ be a feasible solution to the whole problem, with minimum total objective $\Obj_{1,T}(\pvec^*_{1,T})$; and
	\item[(ii)] for every subproblem $(i,j)\in\mathcal{S}$, let $\qvec^*_{i,j}=q^*_i,q^*_{i+1},\ldots,q^*_j$ denote a feasible subproblem solution with minimum partial objective $\Obj_{i,j}(\qvec^*_{i,j})$; and
	\item[(iii)] $\qvec^*_{1,T}=\seq{q^*}{T}$ denote the concatenation of all the subproblem solutions $\qvec^*_{i,j}$, in  order, to form a (not necessarily feasible) solution to the whole problem $(1,T)$, 
\end{remunerate}
then
\begin{equation}
\label{eq:lem4}
\Obj_{1,T}(\qvec^*_{1,T}) 
=   \sum_{(i,j)\in\mathcal{S}} \Obj_{i,j}(\qvec^*_{i,j})
\le \sum_{(i,j)\in\mathcal{S}} \Obj_{i,j}(\pvec^*_{i,j})
= \Obj_{1,T}(\pvec^*_{1,T}). 
\end{equation}
\end{lemma}
\begin{proof}
Follows by recalling~\eqref{eq:totalobjsum} and by noting that for every $(i,j)$, $\Obj_{i,j}(\qvec^*_{i,j})\le\Obj_{i,j}(\pvec^*_{i,j})$, because (except at $i=1$ and $j=T$) minimization of the RHS is subject to the extra constraints $p^*_{i-1}\le p^*_i$ and  $p^*_j\le p^*_{j+1}$. 
\qquad\end{proof}

\begin{lemma}
\label{lemma2a}
For a given RBPSR and for a given partitioning, $\mathcal{S}$, of $(1,T)$ into subproblems, let $\pvec^*_{1,T}=\seq{p^*}{T}$ be a feasible solution to the whole problem, with minimum total objective $\Obj_{1,T}(\pvec^*_{1,T})$; and let $\qvec_{1,T}=\seq{q}{T}$ be any feasible solution to the whole problem, with total objective $\Obj_{1,T}(\qvec_{1,T})$. Then 
\begin{equation}
\label{eq:lem5}
\Obj_{1,T}(\qvec_{1,T}) = \sum_{(i,j)\in\mathcal{S}} \Obj_{i,j}(\qvec_{i,j})
\ge \Obj_{1,T}(\pvec^*_{1,T}). 
\end{equation}
\end{lemma}
\begin{proof}
Follows directly from~\eqref{eq:totalobjsum} and the premise.  
\qquad\end{proof}

\begin{theorem}
\label{theorem1}
Let $\qvec^*_{1,T}=\seq{q^*}{T}$ be a \emph{feasible} solution for $(1,T)$ and let $\mathcal{S}$ be a partitioning of $(1,T)$ into subproblems, such that for every $(i,j)\in\mathcal{S}$, the subsequence $\qvec^*_{i,j}=q^*_i,q^*_{i+1},\ldots,q^*_j$ is the optimal solution to subproblem $(i,j)$, then $\qvec^*_{1,T}$ is the optimal solution to the whole problem $(1,T)$.
\end{theorem}
\begin{proof}
The premises make lemmas~\ref{lemma1a} and~\ref{lemma2a} applicable, for every RBPSR. Since both inequalities~\eqref{eq:lem4} and~\eqref{eq:lem5} are satisfied, $\Obj_{1,T}(\qvec^*_{1,T}) = \Obj_{1,T}(\pvec^*_{1,T})$, where $\pvec^*_{1,T}$ is an optimal solution for each RBPSR. Hence $\qvec^*_{1,T}$ is optimal for every RBPSR and is by corollary~\ref{corollary2} the unique optimal solution.  
\qquad\end{proof}

\subsection{Constant subproblem solutions}
In what follows constant subproblem solutions will be of central importance. A solution $\pvec_{i,j}$ is constant if $p_i=p_{i+1}=\cdots=p_j=q$, for some $0\le q\le1$. In this case, we use the short-hand notation $\Obj_{i,j}(q)=\Obj_{i,j}(\pvec_{i,j})$ to denote the subproblem objective, and this may be expressed as: 
\begin{equation}
\label{eq:objq}
\begin{split}
\Obj_{i,j}(q)=\Obj_{i,j}(\pvec_{i,j})&=\sum_{t=i}^j \wsf(\ell_t) \Crho(\ell_t,q)\\
&=m v_1\Crho(\prop_1,q) +n v_2\Crho(\prop_2,q),
\end{split}
\end{equation}
where $m$ is the number of $\prop_1$-labels and $n$ the number of $\prop_2$-labels. Note:
\begin{remunerate}
	\item A constant subproblem solution is always \emph{feasible}.
	\item If it exists, the optimal solution to an arbitrary subproblem may or may not be constant. 
\end{remunerate}
Whether optimal or not, it is important to examine the behaviour of subproblem solutions that are constrained to be constant. This behaviour is governed by the quasiconvex\footnote{A real-valued function $f(p)$, defined on a real interval is \emph{quasiconvex}, if every sublevel set of the form $\set{p|f(p)<a}$ is convex (i.e. a real interval)~\cite{avriel}. Lemma~\ref{lemma3} shows that $\Obj_{i,j}(q)$ is \emph{quasiconvex}.} properties of $\Obj_{i,j}(q)$ as summarized in the following lemma:
\begin{lemma}
\label{lemma3}
Let $r_{i,j}=\frac{v_1m}{v_1m+v_2n}$, where $m$ is the number of $\prop_1$-labels and $n$ the number of $\prop_2$-labels in the subsequence $(i,j)$, and let $\Obj_{i,j}(q)=m v_1\Crho(\prop_1,q)+nv_2\Crho(\prop_2,q)$ be the objective for the constant subproblem solution, $p_i=p_{i+1}=\cdots=p_j=q$, then the following properties hold, where $\Crho$ is \emph{any} RBPSR, and where we also note the specialization for strict RBPSR's:
\begin{remunerate}
	\item\label{lem1prop1} If $q\le q'\le r_{i,j}$, then $\Obj_{i,j}(q)\ge \Obj_{i,j}(q')\ge\Obj_{i,j}(r_{i,j})$.
\begin{description}
  \item[strict case:] If $q<q'\le r_{i,j}$, then $\Obj_{i,j}(q)> \Obj_{i,j}(q')$.
\end{description}
 	\item\label{lem1prop2} If $q'\ge q\ge r_{i,j}$, then $\Obj_{i,j}(q')\ge \Obj_{i,j}(q)\ge\Obj_{i,j}(r_{i,j})$.
\begin{description}
  \item[strict case:] If $q'>q\ge r_{i,j}$, then $\Obj_{i,j}(q')> \Obj_{i,j}(q)$.
\end{description}
	\item\label{lem1prop3} $\min_q\Obj_{i,j}(q) = \Obj_{i,j}(r_{i,j})$, 
\begin{description}
  \item[strict case:] $q=r_{i,j}$ is the unique minimum.
\end{description}
(This is the \emph{salient property} of binary proper scoring rules, which was mentioned above.)
\end{remunerate}
\end{lemma}

\begin{proof}
For convenience in this proof, we drop the subscripts $i,j$, letting $r=r_{i,j}=\frac{mv_1}{mv_1+nv_2}$. The expected value of $\Crho(\prop,q)$ w.r.t. probability $r$ is: 
\begin{equation}
\label{eq:eq}
\begin{split}
e(q) &= \Exs{\prop}{r}{\Crho(\prop,q)} 
= \tfrac{1}{mv_1+nv_2}\Obj_{i,j}(q)\\ 
&= r \Crho(\prop_1,q)+(1-r)\Crho(\prop_2,q)
\end{split}
\end{equation}
Clearly, if the above properties hold for $e(q)$, then they will also hold for $\Obj_{i,j}(q)$. We prove these properties for $e(q)$ by letting $q \le q'$ and by examining the sign of $\Delta_e = e(q')-e(q)$: If $q'=q$, then $\Delta_e=0$. If $q<q'$, then~\eqref{eq:c_rbpsr} gives:
\begin{equation}
\label{eq:Delta_e}
\Delta_e = \int_q^{q'} (\eta-r) \frac{\rho(\eta)}{\eta(1-\eta)} \,d\eta
\end{equation}
The non-strict versions of properties 1,2 and 3 now follow from the following observation: Since $\rho(\eta)\ge0$ for $0\le \eta \le 1$, the sign of the integrand and therefore of $\Delta_e$ depends solely on the sign of $(\eta-r)$, giving:
\begin{romannum}
	\item $\Delta_e \ge 0$, if $r\le q<q'$.
	\item $\Delta_e \le 0$, if $q<q'\le r$.
\end{romannum}
If more specifically, $\rho(\eta)>0$ \emph{almost everywhere}, then for any $0\le q < q' \le 1$, we have $|\Delta_e|>0$. In this case, the RBPSR is denoted \emph{strict} and we have:
\begin{romannum}
	\item $\Delta_e > 0$, if $r\le q<q'$.
	\item $\Delta_e < 0$, if $q<q'\le r$. 
\end{romannum}
which concludes the proof also for the strict cases.\qquad\end{proof}

For now, we need only  property~\ref{lem1prop3} to proceed. We use the other properties later. The optimal constant subproblem solution is characterized in the following theorem:
\begin{theorem}
\label{theorem2}
If the optimal solution to subproblem $(i,j)$ is constant, then: 
\begin{remunerate}
	\item \label{item:th1_it1} The constant is $r_{i,j}$. 
  \item \label{item:th1_it2} For any index $k$, such that $i\le k\le j$, the following are both true: 
\begin{romannum}
	\item \label{item:th1_it2a}$r_{i,k} \ge r_{i,j}$
	\item \label{item:th1_it2b}$r_{k,j} \le r_{i,j}$
\end{romannum}
where $r_{i,k}$ and $r_{k,j}$ are defined in a similar way to $r_{i,j}$, but for the subproblems $(i,k)$ and $(k,j)$. 
\end{remunerate}
\end{theorem}

\begin{proof}
Property~\ref{item:th1_it1} of this theorem follows directly from property~\ref{lem1prop3} of lemma~\ref{lemma3}.
To prove property~\ref{item:th1_it2}, we use contradiction: If the negation of~2(i) were true, namely $r_{i,k} < r_{i,j}$, then the non-constant solution $p_i=\cdots=p_k=r_{i,k}<p_{k+1}=\cdots=p_j=r_{i,j}$ would be feasible and (by property~\ref{lem1prop3} of lemma~\ref{lemma3}) would have lower objective, namely $\Obj_{i,k}(r_{i,k})+\Obj_{k+1,j}(r_{i,j})$, for any strict RBPSR, than that of the constant solution, namely $\Obj_{i,k}(r_{i,j})+\Obj_{k+1,j}(r_{i,j})$. This contradicts the premise that the optimal solution is constant, so that~\ref{item:th1_it2}(i) must be true. Property~2(ii) is proved by a similar contradiction. 
\qquad\end{proof}

\subsection{Pooling adjacent constant solutions}
This section shows (using lemmas~\ref{lemma4} and~\ref{lemma5} to prove theorem~\ref{theorem3}) when and how optimal constant subproblem solutions may be assembled by pooling smaller adjacent constant solutions:
\begin{lemma}
\label{lemma4}
Given a subproblem $(i,j)$, for which the optimal solution is constant (at $r_{i,j}$), we can form the \emph{augmented subproblem}, with the additional constraint that the solution at $j$ must satisfy $p_j\le\alpha$, for some $\alpha$ such that $0\le\alpha<r_{i,j}$. That is, the solution to the augmented subproblem must satisfy $0\le p_i\le p_{i+1}\le\cdots\le p_j\le\alpha<r_{i,j}$. Then the augmented subproblem solution is  optimized, for every RBPSR, by the constant solution $p_i=p_{i+1}=\cdots=p_j=\alpha$.
\end{lemma}
\begin{proof}
Feasible solutions to the augmented subproblem must satisfy either (i) $p_i=\cdots=p_j=\alpha$, or (ii) $p_i<\alpha$. We need to show that there is no feasible solution of type (ii), which has a lower objective value, for any RBPSR, than solution (i). 

For a given solution, let $k$ be an index such that $i\le k\le j$ and $p_i=p_{i+1}=\cdots=p_k$. By combining the premises of this lemma with property~2(i) of theorem~\ref{theorem2}, we find: $p_i=\cdots=p_k \le \alpha < r_{i,j}\le r_{i,k}$, or more succinctly: $p_i=\cdots=p_k \le \alpha < r_{i,k}$. Now the monotonicity property~\ref{lem1prop1} of lemma~\ref{lemma3} shows that the value of $p_i=\cdots=p_k$, which is optimal for all BPSRs must be as large as allowed by the constraints. This means if we start at $k=i$, then $p_i$ is optimized at the constraint $p_i=p_{i+1}$. Next we set $k=i+1$ to see that $p_i=p_{i+1}$ is optimized at the next constraint $p_i=p_{i+1}=p_{i+2}$. We keep incrementing $k$, until we find the optimum for the augmented subproblem at the constant solution $p_i=\cdots=p_j=\alpha$.
\qquad\end{proof}

\begin{lemma}
\label{lemma5}
Given a subproblem $(i,j)$, for which the optimal solution is constant (at $r_{i,j}$), we can form the \emph{augmented subproblem}, with the additional constraint that the solution at $i$ must satisfy $\alpha\le p_i$, for some $\alpha$ such that $r_{i,j}\le\alpha\le1$. That is, the solution to the augmented subproblem must satisfy $r_{i,j}<\alpha\le p_i\le p_{i+1}\le\cdots\le p_j \le 1$. Then the augmented subproblem solution is optimized, for every RBPSR, by the constant solution $p_i=p_{i+1}=\cdots=p_j=\alpha$.
\end{lemma}
\begin{proof}
The proof is similar to that of lemma~\ref{lemma4}, but here we invoke property~2(ii) of theorem~\ref{theorem2}, to find: $r_{k,j}< \alpha \le p_k=\cdots=p_j$ and we use the monotonicity property~\ref{lem1prop2} of lemma~\ref{lemma3} to show that the value of $p_k=\cdots=p_j$, which is optimal for all RBPSR's, must be as small as allowed by the constraints. 
\qquad\end{proof}

\begin{theorem}
\label{theorem3}
Given indices $i\le k\le j$ such that the optimal subproblem solutions for the two \emph{adjacent} subproblems, $(i,k)$ and $(k+1,j)$, are both constant and therefore (by theorem~\ref{theorem2}) have the respective values $r_{i,k}$ and $r_{k+1,j}$, then, whenever $r_{i,k} \ge r_{k+1,j}$, the optimal solution for the pooled subproblem $(i,j)$ is also constant, and has the value $r_{i,j}$. 
\end{theorem}

\begin{proof}
First consider the case $r_{i,k}=r_{k+1,j}$. Since this forms a constant solution to subproblem $(i,j)$, by theorem~\ref{theorem2}, the optimal solution is $r_{i,j}$.
  
Next consider $r_{i,k}>r_{k+1,j}$. The solution $p_i=\cdots=p_k=r_{i,k}>p_{k+1}=\cdots=p_j=r_{k+1,j}$ is not feasible. A feasible solution must obey $p_k\le\alpha\le p_{k+1}$, for some $\alpha$. There are three possibilities for the value of $\alpha$: (i) $\alpha\le r_{k+1,j}$; (ii) $r_{k+1,j}<\alpha<r_{i,k}$; or (iii) $r_{i,k}\le\alpha$. We examine each in turn:
\begin{remunerate}
	\item[(i)] If $\alpha\le r_{k+1,j}<r_{i,k}$, then the left subproblem $(i,k)$ is augmented by the constraint $\alpha <r_{i,k}$, so that lemma~\ref{lemma4} applies and it is optimized at the constant solution $\alpha$, while the right subproblem $(k+1,j)$ is not further constrained and is still optimized at $r_{k+1,j}$. We can now optimize the total solution for $(i,j)$ by adjusting $\alpha$: By the monotonicity property~\ref{lem1prop1} of lemma~\ref{lemma3}, the left subproblem objective and therefore also the total objective for $(i,j)$ is optimized at the upper boundary $\alpha=r_{k+1,j}$. In other words, in this case, the optimum for subproblem $(i,j)$ is a constant solution.
	\item[(ii)] If $r_{k+1,j}<\alpha<r_{i,k}$, then lemma~\ref{lemma4} applies to the left subproblem and lemma~\ref{lemma5} applies to the right subproblem, so that both subproblems and therefore also the total objective for $(i,j)$ are all optimized at $\alpha$. In this case also we have a constant solution for $(i,j)$. 
	\item[(iii)] If $r_{k+1,j}<r_{i,k}\le\alpha$, then the right subproblem is augmented while the left subproblem is not further constrained. We can now use lemma~\ref{lemma5} and property~\ref{lem1prop2} of lemma~\ref{lemma3}, in a similar way to case (i) to show that in this case also, the optimum solution is  constant.
\end{remunerate}
Since the three cases exhaust the possibilities for choosing $\alpha$, the optimal solution is indeed constant and by theorem~\ref{theorem2} the optimum is at $r_{i,j}$. 
\qquad\end{proof}

\subsection{The PAV algorithm}
\label{sec:pav_alg}
We can now use theorems~\ref{theorem1}, \ref{theorem2} and~\ref{theorem3} to construct a proof that a version of the \emph{pool-adjacent-violators} (PAV) algorithm solves the whole problem $(1,T)$.
\begin{theorem}
The PAV algorithm solves the problem stated in~\S\ref{sec:main_problem}. 
\end{theorem}
\begin{proof}
The proof is constructive. The strategy is to satisfy the conditions for theorem~\ref{theorem1}, by starting with optimal constant subproblem solutions of length 1 and then to iteratively combine  them via theorem~\ref{theorem3}, into longer optimal constant solutions until the total solution is feasible. The algorithm proceeds as follows:
  
\textbf{input:}
\begin{romannum}
  \item labels, $\seq{\ell}{T}\in\set{\prop_1,\prop_2}$.
  \item weights, $v_1,v_2>0$.
\end{romannum}

\textbf{variables:}
\begin{romannum}
	\item $\mathcal{S}$, a partitioning of problem $(1,T)$ into adjacent, non-overlapping subproblems.
  \item $\qvec^*_{1,T}=\seq{q^*}{T}$, a tentative (not necessarily feasible) solution for problem $(1,T)$.
\end{romannum}

\textbf{loop invariant:}
For every subproblem $(i,j)\in\mathcal{S}$: 
\begin{romannum}
\item The optimal subproblem solution is constant.
\item The partial solution $\qvec^*_{i,j}=q^*_i,q^*_{i+1},\ldots,q^*_j$ is equal to the optimal subproblem solution, i.e. constant, with value $r_{i,j}$ (by theorem~\ref{theorem2}).
\end{romannum}

\textbf{initialization:}
Let $\mathcal{S}$ be the finest partitioning into subproblems, so that there are $T$ subproblems, each spanning a single index. Clearly every subproblem $(i,i)$ has a constant solution, optimized at $q^*_i=r_{i,i}$, which is $1$, if $\ell_t=\prop_1$, or $0$, if $\ell_t=\prop_2$. This initial solution $\qvec^*_{1,T}$ respects the loop invariant, but is most probably not feasible. 

\textbf{iteration:}
While $\qvec^*_{1,T}$ is not feasible:
\begin{remunerate}
	\item Find any pair of adjacent subproblems, $(i,k),(k+1,j)\in\mathcal{S}$, for which the solutions are equal or violate monotonicity: $r_{i,k}\ge r_{k+1,j}$. 
	\item Pool $(i,k)$ and $(k+1,j)$ into one subproblem $(i,j)$, by adjusting $\mathcal{S}$ and by assigning the constant solution $r_{i,j}$ to $\qvec^*_{i,j}$, which by theorem~\ref{theorem3} is optimal for $(i,j)$, thus maintaining the loop invariant.
\end{remunerate}

\textbf{termination:}
Clearly the iteration must terminate after at most $T-1$ pooling steps, at which time $\qvec^*_{1,T}$ is now feasible and is still optimal for every subproblem. By theorem~\ref{theorem1}, $\qvec^*_{1,T}$ is then the unique optimal solution to problem $(1,T)$. 
\qquad\end{proof}
    
\section{The PAV-LLR algorithm}
\label{sec:PAV-LLR}
The PAV algorithm as presented above finds solutions in the form of \emph{probabilities}. Here we show how to use it to find solutions in terms of \emph{log-likelihood-ratios}. It will be convenient here to express Bayes' rule in terms of the logit function, $\logit(p)=\log\frac{p}{1-p}$. Note logit is a monotonic rising bijection between $[0,1]$ and the extended real line. Its inverse is the sigmoid function, $\sigma(w)=\frac{1}{1+e^{-w}}$. Bayes' rule is now~\cite{PTTLOS}:
\begin{equation}
\label{eq:br}
  \logit P(\prop_1|s_t) = w_t +\pi 
\end{equation} 
where the LHS is the \emph{posterior log-odds}, $w_t = \log\frac{P(s_t|\prop_1)}{P(s_t|\prop_2)}$ is the \emph{log-likelihood-ratio}, and $\pi=\logit P(\prop_1)$ is the \emph{prior log-odds}.

The problem that is solved by the PAV-LLR algorithm can now be described as follows:
\begin{remunerate}
\item There is given:
\begin{romannum}
	\item Labels, $\seq{\ell}{T}\in\set{\prop_1,\prop_2}$. We denote as $T_1$ and $T_2$ the respective numbers of $\prop_1$ and $\prop_2$ labels in this sequence, so that $T_1+T_2 = T$.
	\item Prior log-odds $\pi$, where $-\infty<\pi<\infty$. This determines a prior probability distribution for the two classes, namely $\bigl(P(\prop_1),P(\prop_2)\bigr)=\bigl(\sigma(\pi),1-\sigma(\pi)\bigr)$, which may be \emph{different} from the label proportions $\bigl(\frac{T_1}{T},\frac{T_2}{T}\bigl)$. 
	\item An RBPSR $\Crho$
\end{romannum}
\item There is required a solution $\nllv_{1,T}=\seq{w}{T}$, which minimizes the following objective:
\begin{align}
\label{eq:obj_pav_llr}
\Obj_{1,T}(\nllv_{1,T}) &= \sum_{t=1}^T \wsf(\ell_t) \Crho(\ell_t,p_t), \\
\label{eq:yabr}
p_t &= \sigma(w_t+\pi), \\
\label{eq:v1}
v_1 &= \wsf(\prop_1) = \frac{\sigma(\pi)}{T_1}
                     = \frac{P(\prop_1)}{T_1},\\  
\label{eq:v2}
v_2 &= \wsf(\prop_2) = \frac{1-\sigma(\pi)}{T_2}
                     = \frac{P(\prop_2)}{T_2} 
\end{align}
(The weights $v_1,v_2$ are chosen thus\footnote{This kind of class-conditional weighting has been used in several formal evaluations of the technologies of \emph{automatic speaker recognition} and \emph{automatic language recognition}, to weight the error-rates of hard recognition decisions~\cite{david,lre_odyssey06} and more recently to also weight logarithmic proper scoring of recognition outputs in log-likelihood-ratio form~\cite{csl,sre_odyssey06,lre_odyssey08}.} to cancel the influence of the proportions of label types, and to re-weight the optimization objective with the given prior probabilities for the two classes, but we show below that this re-weighting is irrelevant when optimizing with PAV.)
\item The minimization is subject to the monotonicity constraint:
\begin{equation}
\label{eq:constr_pav_llr}
-\infty \le w_1\le w_2 \le \cdots\le w_T \le \infty,
\end{equation} 
which by the monotonicity of~\eqref{eq:br} and the logit transformation is equivalent to~\eqref{eq:mon_constr}.
\end{remunerate}
This problem is solved by first finding the probabilities $\seq{p}{T}$ via the PAV algorithm and then inverting~\eqref{eq:yabr} to find $w_t=\logit(p_t)-\pi$. We already know that the solution is independent of the RBPSR, but remarkably, it is \emph{also} independent of the prior $\pi$. This is shown in the following theorem: 
\begin{theorem}
\label{th:pav-lr}
Let $p_t=\PAV_t\bigl((\seq{\ell}{T}),(v_1,v_2)\bigr)$ be given by~\eqref{eq:pav_maxmin}, then the problem of minimizing objective~\eqref{eq:obj_pav_llr}, subject to monotonicity constraint~\eqref{eq:constr_pav_llr} has the unique solution:
\begin{equation}
w_t=\logit \PAV_t\bigl((\seq{\ell}{T}),(1,1)\bigr)
-\logit\frac{T_1}{T}
\end{equation}  
This solution is simultaneously optimal for every RBPSR, $\Crho$, and any prior log-odds, $-\infty<\pi<\infty$.
\end{theorem}
\begin{proof}
By the properties of the PAV as proved in~\S\ref{sec:pav_alg} and since logit is a strictly monotonic rising bijection, it is clear that for all RBPSR's and for a given $\pi$, this minimization is solved as
\begin{equation}
w_t=\logit \PAV_t\bigl((\seq{\ell}{T}),(v_1,v_2)\bigr)
-\pi
\end{equation}
where $\pi$ determines $v_1$ and $v_2$ via~\eqref{eq:v1} and~\eqref{eq:v2}. By corollary~\ref{corollary2}, we can write component $t$ of this solution, in closed form:
\begin{equation}
\begin{split}
w_t &= \logit\left(\max_{1\le i\le t} \, \min_{t\le j\le T} r_{i,j}\right) - \pi \\
&= \max_{1\le i\le t} \, \min_{t\le j\le T} \logit r_{i,j} - \pi
\end{split}
\end{equation} 
Now observe that:
\begin{equation}
\begin{split}
\logit r_{i,j} &= \logit \frac{v_1m_{i,j}}{v_1m_{i,j}+v_2n_{i,j}}\\
&= \logit\frac{m_{i,j}}{m_{i,j}+n_{i,j}} -\logit \frac{T_1}{T} + \pi,
\end{split}
\end{equation}
which shows that $w_t$ is independent of $\pi$. Now the prior may be conveniently chosen to equal the label proportion, $\pi=\logit\frac{T_1}{T}$, to give an un-weighted PAV, with $v_1=v_2=1$.
\qquad\end{proof}

\section{Discussion}
\label{discussion}
We have shown that the problem of monotonic, non-parametric calibration of binary pattern recognition scores is optimally solved by PAV, for all regular binary proper scoring rules. This is true for calibration in posterior probability form and also in log-likelihood-ratio form. 

We conclude by addressing some concerns that readers may have about whether the optimization problem solved here is actually useful in real pattern recognition practice, where a calibration transform is trained in a supervised way (as here) on some training data, but is then utilized later on \emph{new unsupervised} data. 

The first concern we address is about the non-parametric nature of the PAV mapping, because for general real scores there will be new unmapped score values. An obvious solution is to map new values by interpolating between the (input,output) pairs in the PAV solution and this was indeed done in several of the references cited in this paper (see e.g.~\cite{synergy} for an interpolation algorithm). 

Another concern is that the PAV mapping from scores to calibrated outputs has flat regions (all those constant subproblem solutions) and is therefore not an invertible transformation. Invertible transformations are information-preserving, but non-invertible transformations may lose some of the relevant information contained in the input score. This concern is answered by noting that expectations of proper scoring rules are generalized information measures~\cite{DeGroot,Dawid} and that in particular the expectation of the logarithmic scoring rule is equivalent to Shannon's cross-entropy information measure~\cite{coverthomas}. So by optimizing proper scoring rules, we are indeed optimizing the information relevant to discriminating between the two classes. Also note that a \emph{strictly} monotonic (i.e. invertible) transformation can be formed by adding an arbitrarily small strictly monotonic perturbation to the PAV solution. The PAV solution can be viewed as the argument of the infimum of the RBPSR objective, over all strictly rising monotonic transformations. 

In our own work on calibration of speaker recognition log-likelihood-ratios~\cite{taslp}, we have chosen to use strictly monotonic rising \emph{parametric} calibration transformations, rather than PAV. However, we then do use the PAV calibration transformation in the supporting role of \emph{evaluating} how well our parametric calibration strategies work. In this role, the PAV forms a well-defined reference against which other calibration strategies can be compared, since it is the best possible monotonic transformation that can be found on a given set of supervised evaluation data. It is in this evaluation role, that we consider the optimality properties of the PAV to be particularly important.       

For details on how we employ PAV as an evaluation tool\footnote{Our PAV-based evaluation tools are available as a free MATLAB toolkit here: \url{http://www.dsp.sun.ac.za/~nbrummer/focal/}}, see~\cite{csl,springer}.

\section*{Acknowledgments}
We wish to thank Daniel Ramos for hours of discussing PAV and calibration, and without whose enthusiastic support this paper would not have been written.

\appendix
\section{Note on RBPSR family}
\label{appendix}
Some notes follow, to place our definition of the RBPSR family, as defined in~\S\ref{sec:rbpsr} in context of previous work. Our \emph{regularity} condition (i), directly below~\eqref{eq:c_rbpsr},  is adapted from~\cite{Dawid,Gneiting}. General families of binary proper scoring rules have been represented in a variety of ways (see~\cite{Gneiting} and references therein), including also integral representations that are very similar (but not identical in form) to our~\eqref{eq:c_rbpsr}. See for example~\cite{DeGFien}, where the form $\int_q^1 \rho'(\eta)\,d\eta, \;\; \int_0^q \frac{\eta}{1-\eta} \rho'(\eta)\,d\eta$ was used; or~\cite{Buja_degrees,Gneiting} where 
$\int_q^1 (1-\eta) \rho''(\eta)\,d\eta, \;\;
\int_0^q \eta \rho''(\eta)\,d\eta$
was used. Equivalence to~\eqref{eq:c_rbpsr} is established by letting $\rho'(\eta)=\frac{\rho(\eta)}{\eta}$ and $\rho''(\eta)=\frac{\rho(\eta)}{\eta(1-\eta)}$. The advantage of the form~\eqref{eq:c_rbpsr} which we adopt here, is that the weighting function $\rho(\eta)$ is always in the form of a normalized probability density, which gives the natural interpretation of \emph{expectation} to these integrals.

The reader may notice that it is easy (e.g. by applying an affine transform to~\eqref{eq:c_rbpsr}) to find a binary proper scoring rule which satisfies the properties of lemma~\ref{lemma3}, but which is not in the family defined by~\eqref{eq:c_rbpsr}. There are however equivalence classes of proper scoring rules, where the members of a class are all equivalent for making minimum-expected-cost Bayes decisions~\cite{DeGroot,Dawid}. Elimination of this redundancy allows normalization of arbitrary proper scoring rules in such a way that the family~\eqref{eq:c_rbpsr} becomes representative for the members of these equivalence classes~\cite{csl}.

\end{document}